\documentclass{IEEEtran}
\usepackage{amsmath,amssymb,amsfonts}
\usepackage{amsthm}
\usepackage{algorithmic}
\usepackage{graphicx}
\usepackage{textcomp}

\usepackage{subfiles}

\theoremstyle{plain}
\newtheorem{theorem}{Theorem}

\newtheorem{lemma}{Lemma}
\theoremstyle{definition}
\newtheorem{assumption}{Assumption}
\newtheorem{definition}{Definition}

\theoremstyle{remark}
\newtheorem{remark}{Remark}

\renewcommand{\epsilon}{\varepsilon}

\newcommand{\lf}{\left}
\newcommand{\rg}{\right}

\newcommand{\ii}{\infty}

\newcommand{\df}{\,d}

\newcommand{\sbs}{\subseteq}

\newcommand{\tms}{\times}

\newcommand{\seq}[2]{\left\{ #1 \right\}_{#2 \in \mathbb{N}}}
\newcommand{\net}[3]{\left( #1 \right)_{#2 \in #3}}

\newcommand{\N}{\mathbb{N}}

\newcommand{\R}{\mathbb{R}}

\newcommand{\ps}{\psi}

\newcommand{\textall}{\quad\text{for all }}

\newcommand{\Rx}{\mathbb{R}^{n}}
\newcommand{\Ru}{\mathbb{R}^{p}}
\newcommand{\Rd}{\mathbb{R}^{q}}

\newcommand{\ddimension}{q}

\newcommand{\tdum}{s}
\newcommand{\tdumdum}{r}
\newcommand{\tfin}{0}

\newcommand{\tval}{t}
\newcommand{\xval}{x}
\newcommand{\yval}{y}
\newcommand{\uval}{a}
\newcommand{\dval}{b}

\newcommand{\uvals}{\mathcal{A}}
\newcommand{\dvals}{\mathcal{B}}

\newcommand{\xsig}{\mathbf{x}}

\newcommand{\usig}{\mathbf{a}}
\newcommand{\dsig}{\mathbf{b}}

\newcommand{\usigs}{\mathfrak{A}(t)}
\newcommand{\dsigs}{\mathfrak{B}(t)}

\newcommand{\dstrat}{\gamma}

\newcommand{\dstrats}{\Gamma(t)}

\newcommand{\trajSymbol}{\mathbf{x}}
\newcommand{\traj}{\trajSymbol_{x, t}^{\mathbf{a}, \mathbf{b}}}
\newcommand{\trajDStrat}{\trajSymbol_{x, t}^{\mathbf{a}, \gamma[\mathbf{a}]}}

\newcommand{\plausibleSet}{C([\tval,\tfin])}
\newcommand{\failureSet}{\mathbb{F}(\tval)}

\newcommand{\BRS}{\mathcal{L}(\tval)}
\newcommand{\BAS}{\mathcal{G}(\tval)}

\newcommand{\payoff}{J}

\newcommand{\valuefunction}{v}

\newcommand{\xy}[2]{
    \lf(#1, #2\rg)}
\newcommand{\zerosig}{\mathbf{0}}

\title{An Update to the Level Set Theorems in Hamilton-Jacobi Reachability Analysis
\author{Dylan Hirsch, William McEneaney, Jaime Fisac, Claire Tomlin, Sylvia Herbert}
\thanks{This material is based upon work supported by the National Science Foundation Graduate Research Fellowship Program under Grant No. DGE-2038238. Any opinions, findings, and conclusions or recommendations expressed in this material are those of the author(s) and do not necessarily reflect the views of the National Science Foundation.}
\thanks{Dylan Hirsch (corresponding author), William McEneaney, and Sylvia Herbert are with the Department of Mechanical and Aerospace Engineering, University of California, San Diego, 9500 Gilman Drive, La Jolla, CA 92093.
        {\tt\small dhirsch@ucsd.edu, wmceneaney@ucsd.edu, sherbert@ucsd.edu.}}%
\thanks{Jaime Fisac is with the Department of Mechanical and Aerospace Engineering, Princeton University, Princeton, NJ 08544. {\tt\small jfisac@princeton.edu.}}%
\thanks{Claire Tomlin is with the Department of Electrical Engineering and Computer Science, University of California, Berkeley, 387 Soda Hall, Berkeley, CA 94720}
}

\begin{document}
\maketitle

\begin{abstract}
Hamilton-Jacobi Reachability (HJR) is an important framework for controlling safety-critical systems despite uncertainty.
Its theoretical underpinnings are rooted in Hamilton-Jacobi Partial Differential Equations, which provide the value function used for controller synthesis.
The Level Set Theorems of HJR allow one to interpret the value function in terms of satisfaction of a qualitative goal (e.g. goal-reaching or obstacle-avoidance).
We here provide a technical update regarding additional criteria needed for these theorems to hold.
\end{abstract}

\section{Introduction}
Hamilton-Jacobi Reachability (HJR) is a powerful framework for controlling safety-critical nonlinear systems \cite{Bansal-Tomlin-HJR-Overview-2017}.
HJR casts certain control tasks, such as ensuring a robot reaches a target or avoids an obstacle, in the framework of differential games.
Theoretical results obtained for these games, together with sophisticated software tools, can then be used to generate sets with liveness and safety guarantees, paired with corresponding optimal controllers \cite{Mitchell-Tomlin-TAC-HJR-2005,Margellos-Lygeros-TAC-Reach-Avoid-2011,Fisac-Sastry-Reach-Avoid-2015}. The same technique can also be used to generate safety filters such as control barrier functions \cite{Choi_2021}.
Recent works have further advanced the utility of HJR using learning techniques \cite{Bansal_2021,Hsu_2021,Akametalu_2024}.

Traditional HJR analyses typically involve three steps.
First, the analyst converts the ``game of kind'' at hand (e.g. do not hit the obstacle) into a ``game of degree'' (e.g. maintain maximal distance from the obstacle).
Next, she solves a Hamilton-Jacobi (HJ) equation for the value of the game associated with this continuous goal.
Finally, she obtains from this value function the set of states from which the user can certify that the system will remain safe.
The theorems used for the final step are known as level set theorems, and several of them exist for various specifications of the game.

The second step is typically the most challenging and demands the bulk of the theoretical attention in the community, as the HJ equations become increasingly complex for more layered tasks (e.g. reach a goal without hitting an obstacle).
However, the third step involves a number of mathematical technicalities that have created some confusion in the community, possibly due to some technical errors in the level set theorems in the literature.

While these technical issues should not be of immediate concern to the practitioner, their theoretical resolution is interesting and nuanced, involving topological arguments.
In this work, we hope to clear this confusion by 1) correcting the technical errors in other level set theorems of HJR, 2) providing a clear and concise general level set theorem that can be applied to all finite horizon games of interest in HJR, and 3) explaining via counterexamples the origin of the confusing technical issues of the previous level set theorems.

Several others have explored topological approaches to the level set theorems.
Indeed, for disturbance-affine systems, we refer the reader to the proofs in \cite{Cardaliaguet-Two-Player-One-Target-1996,Gammoudi_2023}.
Our own approach builds upon the theory established in \cite{Varaiya-Existence-SIAM-1967}, which does not assume disturbance-affine dynamics.

With regard to the counterexamples, the two common issues seen in level set theorems in the literature involve the lack of a certain convexity condition on the dynamics and the specification of a target set for the controller player as closed rather than open.
We detail how both of these factors can result in discrepancies when obtaining backward reach and avoid tubes from the value function.

These peculiarities are related to the semi-permeable barriers studied by Isaacs \cite{Isaacs-Differential-Games}, typically for infinite-horizon pursuit-evasion games such as the homicidal chauffeur, and the leadership domains of viability theory \cite{Aubin-Viability-Theory}.
In this work, however, we specifically focus on fixed finite-horizon games with general dynamics, general winning criteria, open-loop control, and non-anticipative disturbance strategies, as is the standard setting in HJR.
\section{Setup}

We consider a dynamical system of the form
\begin{equation}
    \dot{\xsig}(\tdum) = f(\xsig(\tdum), \usig(\tdum), \dsig(\tdum)) \label{eqn:dynamics},
\end{equation}
where $f: \Rx \tms \uvals \tms \dvals \to \Rx$, with $\uvals \sbs \Ru$ and $\dvals \sbs \Rd$.
Here, $\usig$ represents the action of a controller player and $\dsig$ represents the action of an adversarial disturbance player.

We consider a game played on a finite time interval with initial time $\tval < 0$ and final time $\tfin$.
Let $\plausibleSet := \{ \xsig:[\tval,\tfin] \to \Rx \mid \xsig \text{ is continuous} \}$
be the set of plausible trajectories of the system, with corresponding norm $\|\xsig\|_\ii := \max_{\tdum \in [\tval,\tfin]} | \xsig(\tdum) |$ (here $|\cdot|$ represents the standard Euclidean norm).
The disturbance player wishes to ensure the actual trajectory is within some set of failure modes $\failureSet \sbs \plausibleSet$, and the controller player wishes to ensure the opposite.

\begin{remark}[Notation]
    We will use the symbols $\xsig$, $\usig$, and $\dsig$ to represent maps from the time interval $[t,0]$ to $\Rx$, $\uvals$, and $\dvals$, respectively.
    We will use the symbols $\xval$, $\uval$, and $\dval$ to represent specific elements of $\Rx$, $\uvals$, and $\dvals$, respectively.
    Additionally, while the initial time $\tval$ will be considered fixed throughout this work, for notational consistency with other HJR literature, we will still explicitly include $\tval$ in the subscripts or arguments of quantities whose definitions involve $\tval$. 
\end{remark}

Throughout this document, we assume the following:
\begin{assumption}\label{assumption:continuity}
The function $f$ is continuous, and the sets $\uvals$ and $\dvals$ are compact.
\end{assumption}

\begin{assumption}\label{assumption:affine-bounded}
There exists a $K > 0$ such that $|f(\xval,\uval,\dval)| \le K (|\xval| + 1)$ for all $\xval \in \Rx$, $\uval \in \uvals$, and $\dval \in \dvals$.
\end{assumption}

\begin{assumption}\label{assumption:locally-lipschitz}
For each $r > 0$ there is an $L_r > 0$ such that $|f(\xval,\uval,\dval) - f(\xval',\uval,\dval)| \le L_r |\xval - \yval|$ for all $\uval \in \uvals$, $\dval \in \dvals$, and $\xval,\xval' \in \Rx$ such that $|\xval| < r$ and $|\yval| < r$.
\end{assumption}

\begin{assumption}\label{assumption:convexity}
    For each $\xval \in \Rx$ and $\uval \in \uvals$, the set $f(\xval,\uval,\dvals) := \lf\{ f(\xval,\uval,\dval) \mid \dval \in \dvals \rg\}$ is convex.
\end{assumption}

\begin{assumption}\label{assumption:closure}
    The set $\failureSet$ is closed in $\plausibleSet$.
    More explicitly, if $\seq{\xsig_n}{n}$ is a sequence in $\failureSet$ and $\|\xsig - \xsig_n\|_\ii \overset{n \to \ii}{\longrightarrow} 0$ for some $\xsig \in \plausibleSet$, then $\xsig \in \failureSet$.
\end{assumption}

\begin{remark}
    Assumptions \ref{assumption:continuity}-\ref{assumption:locally-lipschitz} are standard and ensure existence and uniqueness of trajectories of \eqref{eqn:dynamics}.
    However, Assumptions \ref{assumption:convexity} and \ref{assumption:closure} are often overlooked, despite generally being needed to characterize the winning sets of each player using the level sets of a corresponding value function.
    Indeed, Assumption \ref{assumption:convexity} is what is missing in the canonical work \cite{Mitchell-Tomlin-TAC-HJR-2005}.
    On the other hand, Assumption \ref{assumption:closure} not being satisfied causes the issue in  \cite{Fisac-Sastry-Reach-Avoid-2015}.
\end{remark}

\begin{remark}
Note that Assumption \ref{assumption:convexity} automatically holds if the dynamics are disturbance-affine and $\dvals$ is convex. 
\end{remark}

\begin{remark}\label{remark:meaning-of-A} To satisfy Assumption \ref{assumption:closure}, one should choose targets to be open sets in $\Rx$ and obstacles to be closed sets in $\Rx$.
    Different choices of the failure modes $\failureSet$ of the system correspond to the various finite-horizon games of interest in HJR.
    For example, consider the game discussed in \cite{Mitchell-Tomlin-TAC-HJR-2005}, in which the controller's goal is to avoid the obstacle $\mathcal{G}_0 \sbs \Rx$ at all times.
    The corresponding $\failureSet$ is the set of all continuous maps $\xsig:[\tval,\tfin] \to \Rx$ such that $\xsig(\tdum) \in \mathcal{G}_0$ for some $\tdum \in [\tval,\tfin]$.
    Analogous choices of $\failureSet$ allow one to apply the upcoming result to games in which the controller's goal is to reach a target by some time, reach a target at some time, or to accomplish either of these goals while avoiding an obstacle.
\end{remark}

\subsection{Signals, strategies, and trajectories}
We let $\usigs$ be the set of Lebesgue measurable maps from $[\tval,\tfin]$ to $\uvals$, and we let $\dsigs$ be the set of Lebesgue measurable maps from $[\tval,\tfin]$ to $\dvals$.
We refer to the elements of $\usigs$ as \textit{control signals} on $[\tval,\tfin]$ and the elements of $\dsigs$ as \textit{disturbance signals} on $[\tval,\tfin]$.
Additionally, we refer to maps from $\usigs$ to $\dsigs$ as \textit{disturbance strategies} on $[\tval,\tfin]$.

\begin{definition}[Non-anticipative disturbance strategy]\label{definition:non-anticipative-strategy}
A disturbance strategy $\dstrat: \usigs \to \dsigs$ on $[\tval, \tfin]$ is said to be \textit{non-anticipative} if for each $\usig_1, \usig_2 \in \usigs$ and each $\tdum \in [\tval,\tfin]$, we have $\dstrat[\usig_1](\tdumdum) = \dstrat[\usig_2](\tdumdum)$ for a.e. $\tdumdum \in [\tval,\tdum]$ whenever $\usig_1(\tdumdum) = \usig_2(\tdumdum)$ for a.e. $\tdumdum \in [\tval,\tdum]$.
\end{definition}
Conceptually, the non-anticipative disturbance strategies represent the ways the disturbance player can react to the controller player, namely only based upon the controller player's decisions at and prior to the present time.
We denote the set of all non-anticipative disturbance strategies on $[\tval,\tfin]$ by $\dstrats$.

\begin{definition}[Carath\'{e}odory solution]\label{definition:caratheodory-solution}
Given $\usig \in \usigs$ and $\dsig \in \dsigs$, we say a function $\xsig:[\tval,\tfin] \to \Rx$ is a \textit{Carath\'{e}odory solution} of \eqref{eqn:dynamics} on $[\tval,\tfin]$ under $\usig$ and $\dsig$ if $\xsig$ is absolutely continuous and \eqref{eqn:dynamics} is satisfied for a.e. $\tdum \in [\tval,\tfin]$.
\end{definition}

Under Assumptions \ref{assumption:continuity}-\ref{assumption:locally-lipschitz}, it is a standard result (see e.g. Theorem 1.2.1 in \cite{Friedman-Differential-Games}) that given $\usig \in \usigs$, $\dsig \in \dsigs$, and $\xval \in \Rx$, there is a unique Carath\'{e}odory solution $\xsig$ of \eqref{eqn:dynamics} on $[\tval,\tfin]$ under $\usig$ and $\dsig$ such that $\xsig(\tval) = \xval$.

We denote by $\traj$ this unique solution.
In other words, $\traj$ is the trajectory of \eqref{eqn:dynamics} corresponding to the initial state $\xval$, initial time $\tval$, control signal $\usig$, and disturbance signal $\dsig$.

\subsection{Description of the game}
Given an initial state $\xval \in \Rx$, first the disturbance player selects a non-anticipative disturbance strategy $\dstrat$ from $\dstrats$.
Thereafter, the controller player selects a control signal $\usig$ from $\usigs$. 
Letting $\dsig := \dstrat[\usig]$ denote the resulting disturbance signal, the disturbance player wins if the state trajectory $\traj$ is in $\failureSet$, and the controller player wins otherwise.

We will refer to the set of initial states from which the controller and disturbance players can be guaranteed to win the game as the ``safe set'' and ``unsafe set,'' denoted $\BRS$ and $\BAS$, respectively.
More precisely,
\begin{align*}
    \BRS &:= \lf\{\xval \in \Rx \mid \forall \dstrat \in \dstrats,~\exists \usig \in \usigs,~\trajDStrat \notin \failureSet \rg\},\\
    \BAS &:= \lf\{\xval \in \Rx \mid \exists \dstrat \in \dstrats,~\forall \usig \in \usigs,~\trajDStrat \in \failureSet \rg\}.
\end{align*}
Note that $\BRS$ and $\BAS$ are indeed complementary.

\section{Level Set Theorem}

The result in this section is a generalized level set theorem that can easily be applied to many sorts of differential games.
After stating the theorem, we demonstrate its application to obstacle avoidance and goal reaching games.

\begin{theorem}\label{theorem:main-theorem}
    Let $\payoff: \plausibleSet \to \R$ be a bounded and continuous functional whose non-strict zero sub-level set is the set of failure modes $\failureSet$, i.e.
    \begin{equation*}
        \failureSet = \{ \xsig \in \plausibleSet \mid \payoff[\xsig] \le 0\}.
    \end{equation*}
    Let the value function $\valuefunction(\cdot, \tval): \Rx \to \R$ be given by
    \begin{equation*}
        \valuefunction(\xval, \tval) = \inf_{\dstrat \in \dstrats} \sup_{\usig\in \usigs} \payoff \lf[ \trajDStrat \rg].
    \end{equation*}
    Then the safe and unsafe sets, $\BRS$ and $\BAS$, are respectively given by
\begin{align}
    \BRS &= \{ \xval \in \Rx \mid \valuefunction(\tval, \xval) > 0\},\label{eqn:main-theorem-brs}\\
    \BAS &= \{ \xval \in \Rx \mid \valuefunction(\tval, \xval) \le 0\}\label{eqn:main-theorem-bas}.
    \end{align}
\end{theorem}

Note that by our choice of norm on $\plausibleSet$, the functional $\payoff: \plausibleSet \to \R$ is continuous if and only if $|\payoff[\xsig] - \payoff[\xsig_n]| \overset{n \to \ii}{\longrightarrow} 0$ whenever $\|\xsig - \xsig_n\|_\ii \overset{n \to \ii}{\longrightarrow} 0$.

\begin{remark}
    In the above theorem, we of course could have assumed instead that $\failureSet$ is the non-strict zero \textit{super}-level set of $\payoff$ and switched the $\inf$ and $\sup$ in the definition of $\valuefunction$.
    In this case the signs in \eqref{eqn:main-theorem-brs} and \eqref{eqn:main-theorem-bas} would switch directions.
    However, due to the fact the disturbance player gets the instantaneous advantage in HJR, there should always be a strict inequality in \eqref{eqn:main-theorem-brs} and a non-strict inequality in \eqref{eqn:main-theorem-bas}.
\end{remark}

\subsection{Interpretation in obstacle avoidance}
Let us now interpret this theorem in the context of the obstacle avoidance game described in Remark \ref{remark:meaning-of-A}. 
If the set $\mathcal{G}_0$ is closed in $\Rx$, then we can choose some bounded, continuous function $g:\Rx \to \R$ whose zero sub-level set is $\mathcal{G}_0$.
Setting $\payoff[\xsig] := \min_{\tdum \in [\tval,\tfin]} g(\xsig(\tdum))$ and letting $\failureSet$ be as in Remark \ref{remark:meaning-of-A}, we see that $\failureSet$ is the zero sub-level set of $\payoff$.
We can then conclude from Theorem \ref{theorem:main-theorem} that for a given $\xval$,
$$\inf_{\dstrat \in \dstrats} \sup_{\usig\in \usigs} \min_{\tdum \in [\tval,\tfin]} g \lf( \trajDStrat(\tdum) \rg) \le 0$$
if and only if there is a non-anticipative disturbance strategy $\dstrat$ that ensures, regardless of $\usig$, $\trajDStrat(\tdum) \in \mathcal{G}_0$ for some $\tdum$.
This is the standard level set theorem for finite-time obstacle avoidance games.

\subsection{Interpretation in goal reaching}
Suppose instead we wish to apply this result in the context of goal reaching.
In particular, let $\mathcal{L}_0 \sbs \Rx$ be an open target set which the controller wishes to achieve prior to final time.
In this case, the set of failure modes $\failureSet$ consists of those continuous maps $\xsig:[\tval, 0] \to \Rx$ for which $\xsig(\tdum) \notin \mathcal{L}_0$ for all $\tdum \in [\tval,\tfin]$.
Note that $\failureSet$ is indeed closed since $\mathcal{L}_0$ is open.

We can choose $\ell:\Rx \to \R$ continuous and bounded such that $\mathcal{L}_0 = \{\xval \in \Rx \mid \ell(\xval) > 0\}$.
Letting $\payoff[\xsig] = \max_{\tdum \in [\tval,0]} \ell(\xsig(\tdum))$, we observe that $\xsig \in \failureSet$ if and only if $\payoff[\xsig] \le 0$.
Theorem \ref{theorem:main-theorem} then implies that for a given $\xval$,
$$\inf_{\dstrat \in \dstrats} \sup_{\usig\in \usigs} \min_{\tdum \in [\tval,\tfin]} \ell \lf( \trajDStrat(\tdum) \rg) \le 0$$
if and only if there is a non-anticipative disturbance strategy $\dstrat$ such that regardless of $\usig$, there is no $\tdum$ for which $\trajDStrat(\tdum) \in \mathcal{L}_0$.
This is the correct level set theorem for finite-time goal reaching games.

\section{Issues with prior level set theorems}

In this section, we demonstrate why Assumptions \ref{assumption:convexity} and \ref{assumption:closure} are needed.
The first counterexample involves an obstacle avoidance problem, and the second counterexample involves a target reaching problem.
We note that in Section 2.2 of \cite{Cardaliaguet-Two-Player-One-Target-1996}, the author discusses a different counterexample showing issues that can arise in an obstacle avoidance problem where the obstacle is an open set (i.e. Assumption \ref{assumption:closure} is violated).

\subsection{Lack of convexity}
The first mistake commonly seen in level set theorems in the literature is a lack of some analogue of Assumption \ref{assumption:convexity}. 
To see why this convexity condition is important, consider again the obstacle avoidance game discussed in Remark \ref{remark:meaning-of-A} and the previous section, where the dynamics are
\begin{align*}
    \dot{\xsig}_1 = \usig + \dsig,\quad\dot{\xsig}_2 = 1,
\end{align*}
with $\uvals = [-1,+1]$, $\dvals = \{-1,+1\}$, and $\mathcal{G}_0 = \lf\{\xy{0}{0}\rg\}$.
Note that Assumption \ref{assumption:convexity} is not satisfied for this system since $\dvals$ is not convex (it consists of only $-1$ or $+1$).

The Hamiltonian corresponding to this game is
$$H(\xval,p) := \max_{\uval \in \uvals} \min_{\dval \in \dvals} p_1 (\uval + \dval) + p_2 = p_2.$$
Letting $g(\xval) = 1 - \exp(-|\xval|^2)$,
one can check that 
\begin{equation*}
    v(x,t) = \begin{cases}
        g(\xval_1,\xval_2 - \tval) & \xval_2 \le t,\\
        g(\xval_1,0) & t \le \xval_2 \le 0,\\
        g(\xval_1,\xval_2) & 0 \le \xval_2,
    \end{cases}
\end{equation*}
is continuously differentiable and satisfies the HJ equation 
\begin{equation*}
    \begin{cases}
        \partial_\tval \valuefunction(\xval,\tval) + \min\lf\{0, H(x, \nabla \valuefunction(\xval,\tval))\rg\} = 0 & \text{in } \Rx \tms \R_{<0},\\
        \valuefunction(\xval,0) = g(\xval) & \text{on } \Rx.
    \end{cases}
\end{equation*}
Because $v\lf(\xy{0}{\tval},\tval\rg) = 0$, the standard level set theorem for obstacle avoidance (Theorem 2 in \cite{Mitchell-Tomlin-TAC-HJR-2005}) then asserts $\xy{0}{\tval} \in \mathcal{G}(\tval)$ for all $\tval < \tfin$.
 
To see that this conclusion is in fact false, fix some initial time $\tval < 0$, set $\xval = \xy{0}{\tval}$, and choose an arbitrary $\dstrat \in \dstrats$.
Letting $\zerosig:[\tval,\tfin] \to \R, \tdum \mapsto 0$ be the zero signal, we can choose some time $\tdum_0 \in (\tval,\tfin]$ at which $\int_\tval^{\tdum_0} \dstrat[\zerosig](\tdumdum) \df \tdumdum \ne 0$ (since otherwise $\dstrat[\zerosig](\tdum) = 0 \notin \dvals$ for a.e. $\tdum \in [\tval,\tfin]$ by the Lebesgue Differentiation Theorem; see Theorem 3.21 in \cite{Folland-Real-Analysis}).

Thus there are two cases: either $\int_\tval^{\tdum_0} \dstrat[\zerosig](\tdumdum) \df \tdumdum > 0$ or $\int_\tval^{\tdum_0} \dstrat[\zerosig](\tdumdum) \df \tdumdum < 0$.
Suppose the first case holds.
Consider the control signal $\usig \in \usigs$ defined by
\begin{equation*}
    \usig(\tdum) = \begin{cases}
         0 & \tdum \in [\tval,\tdum_0) \\
         1 & \tdum \in [\tdum_0,\tfin].
    \end{cases}
\end{equation*}
For each $\tdum \in [\tval,\tfin]$, we write $\trajSymbol_{\xval,\tval}^{\usig,\dstrat[\usig]}(\tdum)$ as $\xy{\xsig_1(\tdum)} {\xsig_2(\tdum)}$.

To reach a contradiction, suppose there is some $\tdum^* \in [\tval, \tfin]$ for which $\xy{\xsig_1(\tdum^*)}{ \xsig_2(\tdum^*)} \in \mathcal{G}_0$.
Since $\xsig_2(\tdum^*) = \tval + \int_\tval^{\tdum^*} 1 \df \tdumdum = \tdum^*$, it follows that $\tdum^* = \tfin$.
But then
\begin{align*}
    \xsig_1(\tdum^*) =& \int_\tval^{\tfin} \usig(\tdumdum) \df \tdumdum + \int_\tval^{\tfin} \dstrat[\usig](\tdumdum) \df \tdumdum\\
    =& \int_{\tdum_0}^{\tfin} 1 \df \tdumdum + \int_\tval^{\tdum_0} \dstrat[\usig](\tdumdum) \df \tdumdum + \int_{\tdum_0}^{\tfin} \dstrat[\usig](\tdumdum) \df \tdumdum\\
    \ge& \int_\tval^{\tdum_0} \dstrat[\usig](\tdumdum) \df \tdumdum = \int_\tval^{\tdum_0} \dstrat[\zerosig](\tdumdum) \df \tdumdum,
\end{align*}
where the final equality follows from non-anticipativity of $\dstrat$.
But then $\trajSymbol_1(\tdum^*) > 0$ by the case assumption, so that $\xy{\xsig_1(\tdum^*)}{\xsig_2(\tdum^*)} \notin \mathcal{G}_0$, providing the desired contradiction.

The other case follows analogously by instead letting $\usig(\tdum) = -1$ for $\tdum > \tdum_0$.
Since $\dstrat$ was chosen arbitrarily, we have shown that for every $\dstrat \in \dstrats$, there is an $\usig \in \usigs$ such that $\trajDStrat(\tdum) \notin \mathcal{G}_0$ for all $\tdum \in [\tval,\tfin]$.
In words, the controller player can ensure that the state will not ever touch $\xy{0}{0}$.

Thus, $\xy{0}{\tval} \notin \BAS$ for all $\tval < \tfin$ even though $v\lf(\xy{0}{\tval}, \tval \rg) = 0$.
By contrast, it is easy to see that for each $\tval < \tfin$, $v\lf(\xy{0}{\tfin}, \tval\rg) = 0$ and $\xy{0}{\tfin} \in \BAS$.
This counterexample demonstrates that for obstacle-avoidance games, interpreting the outcome of the game for points on the zero level set is not straightforward when Assumption \ref{assumption:convexity} does not hold.
For some points on the zero level set, the controller player will win, whereas for others the disturbance player will win.

\subsection{Closed target sets}
A different problem occasionally arises in the literature on the target reaching games introduced in the previous section.
We have previously asserted that targets should be open sets and obstacles should be closed sets. 
Intuitively, this stems from the fact that the disturbance has the instantaneous advantage in the game, creating an asymmetry between the players.

Sometimes, however, versions of the level set theorem for target reaching games assume that the target $\mathcal{L}_0$ is closed and claim that the winning set for the controller is the \textit{non-strict} zero super-level set of the corresponding value function.
We provide a counterexample to this.
Consider a system with the same dynamics as in the last example, but with $\uvals = \dvals = [-1,+1]$.
Let $\mathcal{L}_0 = \lf\{\xy{0}{0}\rg\}$ and define $\ell(\xval) = \exp(-|\xval|^2) - 1$.

The Hamiltonian corresponding to this game is again
$$H(\xval,p) := \max_{\uval \in \uvals} \min_{\dval \in \dvals} p_1 (\uval + \dval) + p_2 = p_2.$$
Letting $\ell(\xval) = \exp(-|\xval|^2) - 1$,
one can check that 
\begin{equation*}
    v(x,t) = \begin{cases}
        \ell(\xval_1,\xval_2 - \tval) & \xval_2 \le t,\\
        \ell(\xval_1,0) & t \le \xval_2 \le 0,\\
        \ell(\xval_1,\xval_2) & 0 \le \xval_2,
    \end{cases}
\end{equation*}
is continuously differentiable and satisfies the HJ equation 
\begin{equation*}
    \begin{cases}
        \partial_\tval \valuefunction(\xval,\tval) + \max\lf\{0, H(x, \nabla \valuefunction(\xval,\tval))\rg\} = 0 & \text{in } \Rx \tms \R_{<0}\\
        \valuefunction(\xval,0) = g(\xval) & \text{on } \Rx.
    \end{cases}
\end{equation*}
The incorrect level set theorem for target reaching (see e.g. Proposition 2 in \cite{Fisac-Sastry-Reach-Avoid-2015}) then asserts $\xy{0}{\tval} \in \mathcal{L}(\tval)$ for all $\tval < \tfin$ since $v\lf( \xy{0}{\tval},\tval \rg) = 0$.

To see this is not the case, fix some $\tval < 0$ and set $\xval = \xy{0}{\tval}$.
Consider the non-anticipative strategy $\dstrat \in \dstrats$ defined by
\begin{equation}\label{eqn:closed-target-counterexample-strategy}
    \dstrat[\usig](\tdum) = 
    \begin{cases} 
    0 & \tdum = \tval \\
    \limsup_{n \to \ii} \chi \lf( \int_t^{t+\frac{1}{n}} \usig(\tdumdum) \df \tdumdum \rg) & \tdum \in (\tval,\tfin],
    \end{cases}
\end{equation}
where $\chi(z) := +1$ for $z \ge 0$ and $\chi(z) := -1$ for $z < 0$.

Note that each disturbance signal $\dstrat[\usig]$ is indeed measurable because it is constant after the initial time.
It can also be checked directly from Definition \ref{definition:non-anticipative-strategy} that this disturbance strategy is indeed non-anticipative.
Intuitively, this is the case because given a control signal $\usig \in \usigs$, the value of the resulting disturbance signal $\dstrat[\usig]$ is $0$ at the initial time and for subsequent times only depends on how $\usig$ behaves infinitesimally close to the initial time. 

Now, fix some arbitrary control signal $\usig \in \usigs$.
For convenience, write $\trajSymbol_{\xval,\tval}^{\usig,\dstrat[\usig]}(\tdum)$ as $\xy{\xsig_1(\tdum)}{\xsig_2(\tdum)}$ for each $\tdum \in [\tval,\tfin]$.
To reach a contradiction, suppose there is some $\tdum^* \in [\tval, \tfin]$ for which $\xy{\xsig_1(\tdum^*)}{ \xsig_2(\tdum^*)} \in \mathcal{L}_0$.
Since $\xsig_2(\tdum^*) = \tval + \int_\tval^{\tdum^*} 1 \df \tdumdum = \tdum^*$, it follows that $\tdum^* = \tfin$.

There are two cases: either 1) for each $n \in \N$ there is some integer $k > n$ such that $\int_{\tval}^{\tval + \frac{1}{k}} \usig(\tdumdum) \df \tdumdum \ge 0$ or 2) $\int_{\tval}^{\tval + \frac{1}{k}} \usig(\tdumdum) \df \tdumdum < 0$ for all sufficiently large integers $k$.

Suppose the first case holds.
Then $\dstrat[\usig](\tdum) = +1$ for all $\tdum \in (\tval, \tfin]$.
Additionally, from the case assumption we can choose some $k \in \N$ for which
$\int_\tval^{\tval+\frac{1}{k}} \usig(\tdumdum) \df \tdumdum \ge 0$.
But then
\begin{align*}
    \xsig_1(\tdum^*) &= \int_\tval^{\tfin} \usig(\tdumdum) \df \tdumdum + \int_\tval^{\tfin} \dstrat[\usig](\tdumdum) \df \tdumdum\\
    &= \int_\tval^{\tval+\frac{1}{k}} \usig(\tdumdum) \df \tdumdum + \int_{\tval+\frac{1}{k}}^{\tfin} \usig(\tdumdum) \df \tdumdum + \int_\tval^{\tfin} 1 \df \tdumdum\\
    &\ge \int_{\tval+\frac{1}{k}}^{\tfin} \usig(\tdumdum) \df \tdumdum + \int_\tval^{\tfin} 1 \df \tdumdum > 0,
\end{align*}
where the final inequality follows from the facts that $\tdum + \frac{1}{k} > \tval$ and $\usig(\tdumdum) \ge -1$.
But then $\xy{\xsig_1(\tdum^*)}{ \xsig_2(\tdum^*)} \notin \mathcal{L}_0$, creating the desired contradiction.

In the second case, we have $\dstrat[\usig](\tdum) = -1$ for all $\tdum \in (\tval,\tfin]$, and we can choose some $k \in \N$ for which $\int_{\tval}^{\tval + \frac{1}{k}} \dstrat[\usig](\tdum) \df \tdum < 0$.
The remainder of the proof follows similarly.

Thus, we have concluded that $\xy{0}{\tval} \notin \BRS$ for all $\tval < 0$ even though $v\lf(\xy{0}{\tval}, \tval\rg) = 0$.
By contrast, it is easy to see that for each $\tval < 0$, $v\lf( \xy{0}{\tfin}, \tval \rg) = 0$ and indeed $\xy{0}{\tfin} \in \BRS$.
This counterexample demonstrates that for target-reaching games, interpreting the outcome of the game for points on the zero level set is not straightforward when the target set $\mathcal{L}_0$ is closed (in which case Assumption \ref{assumption:closure} cannot hold).
For some points on the zero level set, the controller player will win, whereas for others the disturbance player will win.

\section{Proof of Main Theorem}

\newcommand{\netindex}{\alpha}
\newcommand{\directedset}{I}
\newcommand{\trajsUnderU}{\mathbb{X}(\xval, \tval; \usig)}
\newcommand{\trajsAll}{\mathbb{X}(\xval, \tval)}
\newcommand{\dplan}{\theta}
\newcommand{\dplans}{\Theta(\xval, \tval)}
\newcommand{\productSpace}{\Pi_{\usig \in \usigs} \trajsUnderU}
\newcommand{\trajDStratStar}{\trajSymbol_{\xval,\tval}^{\usig, \dstrat^*[\usig]}}
\newcommand{\trajDStrataStar}{\trajSymbol_{\xval,\tval}^{\usig^*, \dstrat[\usig^*]}}

Throughout this section, as previously, we assume Assumptions \ref{assumption:continuity}-\ref{assumption:convexity} all hold.
We build to the proof of Theorem \ref{theorem:main-theorem}.
\begin{definition}[strict lexicographic order]
    The \textit{strict lexicographic order} on $\Rd$ is the strict total order $\ll$ for which given any two elements $a = (a_1,\dots,a_\ddimension)$ and $a' = (a_1',\dots,a_\ddimension')$ of $\Rd$, we have that $a \ll a'$ iff there is a $k \in \{1,\dots,\ddimension\}$ such that $a_i = a_i'$ for each $i \in \{1,\dots,k-1\}$ and $a_k < a_k'$.
\end{definition}

We will make use of the following version of Filippov's Lemma (see Lemma 2.4.2 in \cite{Friedman-Differential-Games}).
The lemma stated below is slightly different from the version in \cite{Friedman-Differential-Games}, but the proof proceeds almost identically.
For completeness, we reproduce the proof below, modifying it where relevant.
\begin{lemma}\label{filippov-lemma}[Filippov's Lemma]
    Suppose $g:[\tval,\tfin] \tms \uvals \tms \dvals \to \Rx$ is continuous.
    Let $\psi: [\tval,\tfin] \to \Rx$ and $\usig:[\tval,\tfin] \to \uvals$ be measurable functions which satisfy $\psi(\tdum) \in g(\tdum, \usig(\tdum), \dvals) := \{g(\tdum,\usig(\tdum),\dval) \mid \dval \in \dvals\}$ for almost every $\tdum \in [\tval,\tfin]$.
    Then there exists a measurable function $\dsig: [\tval,\tfin] \to \dvals$ such that $\psi(\tdum) = g(\tdum,\usig(\tdum),\dsig(\tdum))$ for almost every $\tdum \in [\tval,\tfin]$.
    Moreover, $\dsig$ can be chosen so that for a.e. $\tdum \in [\tval,\tfin]$, $\dsig(\tdum)$ is the minimum element of the set $\{\dval \in \dvals \mid \psi(\tdum) = g(\tdum, \usig(\tdum), \dval)\}$ under the lexicographic order on $\Rd$.
\end{lemma}
\begin{proof}
Without loss of generality, we may assume that $\ps(\tdum) \in g(\tdum, \usig(\tdum), \dvals)$ for every $\tdum \in [\tval,\tfin]$.
We construct $\dsig:[\tval,\tfin] \to \dvals$ as follows.
Given some $\tdum \in [\tval,\tfin]$, let $\dvals_0(\tdum)$ be the set of all $\dval \in \dvals$ for which $\ps(\tdum) = g(\tdum,\usig(\tdum),\dval)$.
Because $g$ is continuous and $\dvals$ is compact, $\dvals_0(\tdum)$ is also compact.
It follows that the set $\dvals_1(\tdum)$ of points in $\dvals_0(\tdum)$ at which the function $h_1: \dvals_0(\tdum) \to \R, (\dval_1,\dots,\dval_{\ddimension}) \mapsto \dval_1$ attains its minimum value is non-empty and compact.
But then the set $\dvals_2(\tdum)$ of points in $\dvals_1(\tdum)$ at which the function $h_2: \dvals_1(\tdum) \to \R, (\dval_1,\dots,\dval_{\ddimension}) \mapsto \dval_2$ attains its minimum value is non-empty and compact as well.
We define $\dvals_3(\tdum),\dots,\dvals_{\ddimension}(\tdum)$ in this fashion and note that $\dvals_{\ddimension}(\tdum)$ must then be a singleton.
We let $\dsig(\tdum)$ be the element of this singleton.
Then in particular, $\dsig(\tdum)$ is the minimum element of $\dvals_0(\tdum)$ under the lexicographic order on $\Rd$.

It remains to show that $\dsig:[\tval,\tfin] \to \dvals$ is measurable, which we show by induction.
Write $\dsig(\cdot)$ as $(\dsig_1(\cdot),\dots,\dsig_{\ddimension}(\cdot))$.
Assume that $\dsig_1,\dots, \dsig_{k-1}$ are each measurable (the proof that $\dsig_1$ is measurable proceeds similarly). 
By Luzin's theorem (for reference, see example Exercise 44 in Section 2.4 of \cite{Folland-Real-Analysis}), for each $r \in \N$ we can choose a closed subset $E_r$ of $[\tval,\tfin]$ with measure no less than $T - t - \frac{1}{r}$ such that $\ps|E_r$, $\usig|E_r$, and $\dsig_1|E_r, \dots, \dsig_{k-1}|E_r$ are all continuous.

We claim that for each $c \in \R$ and $r \in \N$ the set $\{\tdum \in E_r \mid \dsig_k(\tdum) \le c\}$ is closed.
Suppose otherwise.
Then there is an $c^* \in \R$, a $r^* \in \N$, a $\tdum^* \in E$, and a sequence $\seq{\tdum_j}{j}$ in $E_{r^*}$ converging to $\tdum^*$ for which
\begin{equation}\label{proof:filippov-lemma-1}
    \dsig_k(\tdum_j) \le c^* < \dsig_k(\tdum^*)\textall j \in \N.
\end{equation}
But since $\dvals$ is compact, there is a subsequence $\seq{\tdum_{{j}_{j'}}}{j'}$ of $\seq{\tdum_j}{j}$ such that $\dsig(\tdum_{j_{j'}}) \to \dval^*$ for some $\dval^* = (\dval_1^*,\dots,\dval_k^*) \in \dvals$.
It follows from \eqref{proof:filippov-lemma-1} that
\begin{equation}\label{proof:filippov-lemma-3}
    \dval^*_k < \dsig_k(\tdum^*).
\end{equation}
Moreover, since $g(\tdum_{j_{j'}}, \usig(\tdum_{j_{j'}}), \dsig(\tdum_{j_{j'}})) = \ps(\tdum_{j_{j'}})$ for all $j' \in \N$,
then by continuity of $g$ on its domain along with continuity of $\ps$, $\usig$, and $\dsig_1,\dots,\dsig_{k-1}$ on $E_{r^*}$, we have
$$g(\tdum^*,\usig(\tdum^*),(\dsig_1(\tdum^*),\dots,\dsig_{k-1}(\tdum^*), \dval^*_k, \dots, \dval^*_\ddimension)) = \ps(\tdum^*),$$
so that $(\dsig_1(\tdum^*),\dots,\dsig_{k-1}(\tdum^*), \dval^*_k, \dots, \dval^*_\ddimension) \in \dvals_0(\tdum^*).$
It follows from the definition of $\dsig(\tdum^*)$ that $\dsig(\tdum^*) \ll (\dsig_1(\tdum^*),\dots,\dsig_{k-1}(\tdum^*), \dval^*_k, \dots, \dval^*_\ddimension)$, so that in particular $\dsig_k(\tdum^*) \le \dval_k^*$, which contradicts \eqref{proof:filippov-lemma-3}.

Thus $\{\tdum \in E_r \mid \dsig_k(\tdum) \le c\}$ is indeed closed for each $r \in \N$.
But then $\dsig_k|E_r$ is measurable for each $r \in \N$, so $\dsig_k|E$ is also measurable, where $E := \cup_{r \in \N} E_r$.
Since the measure of $E$ is then necessarily $T - t$, it follows that $\dsig_k$ is in fact measurable on $[\tval,\tfin]$.
By induction, we then have that $\dsig$ is measurable on $[\tval,\tfin]$, completing the proof.
\end{proof}

Before we prove Theorem \ref{theorem:main-theorem}, we will need to do some preliminary work.
For each $\xval \in \Rx$, we let $\trajsAll = \{ \traj \mid \usig \in \usigs, \dsig \in \dsigs\}$.
Conceptually, $\trajsAll$ represents the set of all trajectories on $[\tval,\tfin]$ the system may take when starting from state $\xval$.
We endow $\trajsAll$ with the uniform topology.
Recall that with this choice of topology, given a net $\net{\trajSymbol_\netindex}{\netindex}{\directedset}$ in $\trajsAll$ (where $\directedset$ is some directed set) and some $\trajSymbol \in \trajsAll$, then $\trajSymbol_\netindex \to \trajSymbol$ iff $\max_{t \in [\tval,\tfin]} \|\trajSymbol_\netindex(t) - \trajSymbol(t)\| \to 0$.

We now introduce a notion similar to a disturbance strategy and its non-anticipativity condition:
\begin{definition}\label{definition:plan}[Disturbance plan]
    Given $\xval \in \Rx$, we say a map $\dplan:\usigs \to \trajsAll$ is a \textit{disturbance plan} on $[\tval,\tfin]$ starting at $\xval$ if for each $\usig \in \usigs$ there is a $\dsig \in \dsigs$ such that
    $\dplan[\usig] = \traj.$
    Additionally, we say $\dplan$ is \textit{non-anticipative} if for each $\usig_1,\usig_2 \in \usigs$ and each $s \in [\tval,\tfin]$, we have $\dplan[\usig_1](\tdumdum) = \dplan[\usig_2](\tdumdum)$ for all $\tdumdum \in [\tval,\tdum]$ whenever $\usig_1(\tdumdum) = \usig_2(\tdumdum)$ for a.e. $\tdumdum \in [\tval,\tdum]$.
\end{definition}
We denote by $\dplans$ the set of all non-anticipative disturbance plans on $[\tval,\tfin]$ starting at $\xval$.
We endow each $\dplans$ with the topology of point-wise convergence.
Recall that with this choice of topology, given a net $\net{\dplan_\netindex}{\netindex}{\directedset}$ in $\dplans$ and some $\dplan \in \dplans$, then $\dplan_\netindex \to \dplan$ iff $\dplan_\netindex[\usig] \to\dplan[\usig]$ for each $\usig \in \usigs$.


\begin{remark}
Let's allow ourselves a brief aside to clarify how one should think about a disturbance plan.
Suppose $\dplan \in \dplans$.
Then for each $\usig \in \usigs$, $\dplan[\usig]$ represents a choice of a trajectory, starting from the state $\xval$ at time $\tval$, that the disturbance player will choose in response to the knowledge that the controller player has selected the signal $\usig$.
Of course, we require that the trajectory that is chosen is one that can actually be achieved via some disturbance signal $\dsig \in \dsigs$.
Conceptually, a disturbance plan is non-anticipative if the trajectory up to some time is entirely determined by the control signal only up to that time.
\end{remark}

It seems natural that there would be a correspondence between the non-anticipative disturbance strategies in $\dstrats$ and the non-anticipative disturbance plans in $\dplans$.
The nature of this correspondence is summarized in the following theorem.
\begin{lemma}[Correspondence between strategies and plans] \label{lemma:correspondence}
    Let $\tval < \tfin$ and $\xval \in \Rx$.
    \item[(i)] Suppose $\dstrat \in \dstrats$, and let $\dplan:\usigs \to \trajsAll$ be given by $\dplan[\usig] = \trajDStrat$.
    Then $\dplan \in \dplans$.
    \item[(ii)] For each $\dplan \in \dplans$, there is a $\dstrat \in \dstrats$ such that $\dplan[\usig] = \trajDStrat$ for all $\usig \in \usigs$.
\end{lemma}

\begin{proof}(i) Let $\dstrat \in \dstrats$.
It is clear that $\dplan$ is a disturbance plan, so it suffices to show that it is also non-anticipative.
Let $\tdum \in [\tval,\tfin]$ and let $\usig_1, \usig_2 \in \usigs$ agree a.e. on $[\tval,\tdum]$.
Then $ \dsig_1 := \dstrat[\usig_1]$ and $\dsig_2 := \dstrat[\usig_2]$ also agree a.e. on $[\tval, \tdum]$.
Letting $\xsig_1 = \trajSymbol_{\tval,\xval}^{\usig_1,\dsig_1}$ and $\xsig_2 = \trajSymbol_{\tval,\xval}^{\usig_2,\dsig_2}$, observe that
$$\dot{\xsig}_1(\tdumdum)
= f\lf(\xsig_1(\tdumdum), \usig_1(\tdumdum), \dsig_1(\tdumdum) \rg) 
= f\lf(\xsig_1(\tdumdum), \usig_2(\tdumdum), \dsig_2(\tdumdum) \rg)$$
for a.e. $\tdumdum \in [\tval,\tdum]$. Moreover, by definition, $\dot{\xsig}_2(\tdumdum) = f\lf(\xsig_2(\tdumdum), \usig_2(\tdumdum), \dsig_2(\tdumdum) \rg)$ for a.e. $\tdumdum \in [\tval,\tdum]$.
Since $\xsig_1(\tval) = \xval = \xsig_2(\tval)$, it follows from uniqueness of Carath\'{e}odory solutions to \eqref{eqn:dynamics} under Assumptions \ref{assumption:continuity}-\ref{assumption:locally-lipschitz} that $\xsig_1$ and $\xsig_2$ agree on $[\tval,\tdum]$.

(ii) Let $\dplan \in \dplans$.
Define the disturbance strategy $\dstrat: \usigs \to \dsigs$ as follows.
For each $\usig \in \usigs$, let $\xsig_\usig = \dplan[\usig].$
Since each $\xsig_\usig$ is a trajectory in $\trajsUnderU$, we have that
$\dot{\xsig}_\usig(\tdum) \in f\lf(\xsig_\usig(\tdum), \usig(\tdum), \dvals \rg)$ for a.e. $\tdum \in [\tval,\tfin]$.
It follows from Lemma \ref{filippov-lemma} that for each $\usig \in \usigs$, we can choose $\dsig_\usig \in \dsigs$ such that $\dsig_\usig(\tdum)$ is the minimum element in the strict lexicographic order on $\Rd$ of the set $\lf\{\dval \in \dvals ~\mid \dot{\xsig}_\usig(\tdum) = f\lf(\xsig_\usig(\tdum), \usig(\tdum), \dval \rg) \rg\}$ for a.e. $\tdum \in [\tval,\tfin]$.
We set $\dstrat[\usig] = \dsig_\usig$ for each $\usig \in \usigs$.

It suffices to show that $\dstrat$ is non-anticipative.
Let $\tdum \in [\tval,\tfin]$, and let $\usig_1, \usig_2 \in \usigs$ agree a.e. on $[\tval,\tdum]$.
Then $\xsig_{\usig_1}$ and $\xsig_{\usig_2}$ agree on $[\tval,\tdum]$.
Thus for a.e. $\tdumdum \in [\tval,\tdum]$
\begin{align*}
&\lf\{\dval \in \dvals \mid \dot{\xsig}_{\usig_1} (\tdumdum) = f\lf(\xsig_{\usig_1}(\tdumdum), \usig_1(\tdumdum), \dval \rg) \rg\} \\
&\quad= \lf\{\dval \in \dvals \mid \dot{\xsig}_{\usig_2}(\tdumdum) = f\lf(\xsig_{\usig_2}(\tdumdum), \usig_2(\tdumdum), \dval \rg) \rg\}.
\end{align*}
But for a.e. $\tdumdum \in [\tval,\tdum]$, $\dsig_{\usig_1}(\tdumdum)$ and $\dsig_{\usig_2}(\tdumdum)$ are the minimum elements (in the strict lexicographic ordering on $\Ru$) of the set on the left and right of the above equation, respectively.
Thus $\dstrat[\usig_1](\tdumdum) = \dsig_{\usig_1}(\tdumdum) = \dsig_{\usig_2}(\tdumdum) = \dstrat[\usig_2](\tdumdum)$ for a.e. $\tdumdum \in [\tval,\tdum]$.
\end{proof}

In the proof of the upcoming lemma we will make use of the following result, which is immediate from Theorem 2.4.2 in \cite{Friedman-Differential-Games} under Assumption \ref{assumption:convexity}.
Thereafter we mostly follow similarly \cite{Varaiya-Existence-SIAM-1967}, modifying the analysis where needed.
\begin{theorem}[Compactness of subsets of trajectories]~ \label{theorem:compactness-trajectories}
    For each $\xval \in \Rx$ and $\uval \in \uvals$, the set $\{ \traj \mid \dsig \in \dsigs \}$ is a compact subset of $\trajsAll$.
\end{theorem}

\begin{lemma}[Compactness of $\dplans$]\label{lemma:compactness-plans}
    The set $\dplans$ is compact for each $\xval \in \Rx$.
\end{lemma}
\begin{proof}
    For each $\usig \in \usigs$, let $\trajsUnderU = \{ \traj \mid \dsig \in \dsigs\}$.
    Then the set of all disturbance plans on $[\tval,\tfin]$ starting at $\xval$ can be written as $\productSpace$ (i.e. the infinite Cartesian product of all the $\trajsUnderU$).
    We endow $\productSpace$ with the topology of point-wise convergence (i.e. the product topology).
    By Tychonoff's Theorem and Theorem \ref{theorem:compactness-trajectories}, $\productSpace$ is compact, so it is sufficient to show that $\dplans$ is closed in $\productSpace$.
    Let $\net{\dplan_\netindex}{\netindex}{\directedset}$ (where $\directedset$ is some directed set) be a net in $\dplans$ which converges in $\productSpace$ to some $\dplan \in \productSpace$.
    We wish to show $\dplan \in \dplans$, i.e. that $\dplan$ is non-anticipative.

    So let $\tdum \in [\tval,\tfin]$ and suppose $\usig_1, \usig_2 \in \usigs$ agree a.e. on $[\tval,\tdum]$.
    For convenience, for each $\netindex \in \directedset$, let $\xsig_{\netindex,1} = \xsig_\netindex[\usig_1]$ and $\xsig_{\netindex,2} = \dplan_\netindex[\usig_2]$.
    Also, let $\xsig_{1} = \dplan[\usig_1]$ and $\xsig_{2} = \dplan[\usig_2]$.
    Then $\xsig_{\netindex,1} = \xsig_{\netindex,2}$ on $[\tval,\tdum]$ for each value of $\netindex$.
    Since $\xsig_{\netindex,1} \to \xsig_1$ and $\xsig_{\netindex,2} \to \xsig_2$, then
    $\xsig_1 = \xsig_2$ on $[\tval,\tdum]$.
    Thus, $\dplan$ is non-anticipative, so $\dplans$ is closed and thus compact.
\end{proof}

\begin{proof}[Proof of Theorem \ref{theorem:main-theorem}]
We prove \eqref{eqn:main-theorem-bas}, and then \eqref{eqn:main-theorem-brs} follows from taking the complement of both sides of \eqref{eqn:main-theorem-bas}.

Fix $\xval \in \Rx$.
First, suppose that $\valuefunction(\xval,\tval) > 0$.
Arbitrarily choose $\dstrat \in \dstrats$. 
By definition of $\valuefunction$, we have $\inf_{\usig \in \usigs} \payoff\lf[\trajDStrat\rg] > 0$, so that $\payoff\lf[\trajDStrataStar\rg] > 0$ for some $\usig^* \in \usigs$.
But then $\trajDStrataStar \notin \failureSet$ by the theorem hypothesis.
Since $\dstrat$ was arbitrary, it follows that $\xval \notin \BAS$.

Instead suppose $\valuefunction(\xval,\tval) \le 0$.
For each $\usig \in \usigs$, let $h_\usig: \dplans \to \R, \dplan \mapsto \payoff[\dplan[\usig]].$
We claim that each $h_\usig$ is continuous.
Indeed, given some $\usig \in \usigs$ and some $\dplan \in \dplans$, if $\net{\dplan_\netindex}{\netindex}{\directedset}$ is a net is $\dplans$ which converges to $\dplan$, then $\dplan_\netindex[\usig]$ converges to $\dplan[\usig]$, so $h_\usig(\dplan_\netindex) = \payoff[\dplan_\netindex[\usig]] \to \payoff[\dplan[\usig]] = h_\usig(\dplan)$.

Now let $H:\dplans \to \R, \dplan \mapsto \sup_{\usig \in \usigs} h_\usig(\dplan)$.
Then $H$ is lower semi-continuous, and since its domain is compact by Lemma \ref{lemma:compactness-plans}, $H$ achieves its minimum.
In other words, there is some $\dplan^* \in \dplans$ such that
\begin{align}\label{proof:main-theorem-1}
    \sup_{\usig \in \usigs}\payoff\lf[\dplan^*[\usig]\rg] = H[\dplan^*] &= \inf_{\dplan \in \dplans} H[\dplan]\nonumber \\
    &= \inf_{\dplan \in \dplans} \sup_{\usig \in \usigs}\payoff\lf[\dplan[\usig]\rg].  
\end{align}
But by Lemma \ref{lemma:correspondence}(i), there is some $\dstrat^* \in \dstrats$ such that
\begin{equation}\label{proof:main-theorem-2}
    \sup_{\usig \in \usigs}\payoff\lf[\trajDStratStar\rg] = \sup_{\usig \in \usigs}\payoff\lf[\dplan^*[\usig]\rg].
\end{equation}
Moreover, by Lemma \ref{lemma:correspondence},
\begin{equation}\label{proof:main-theorem-3}
    \inf_{\dplan \in \dplans} \sup_{\usig \in \usigs}\payoff\lf[\dplan[\usig]\rg] = \inf_{\dstrat \in \dstrats} \sup_{\usig \in \usigs}\payoff\lf[\trajDStrat\rg].
\end{equation}
Combining \eqref{proof:main-theorem-1}-\eqref{proof:main-theorem-3} gives 
$$\sup_{\usig \in \usigs}\payoff\lf[\trajDStratStar\rg] = \inf_{\dstrat \in \dstrats} \sup_{\usig \in \usigs}\payoff\lf[\trajDStrat\rg] = \valuefunction(\xval, \tval) \le  0.$$
Thus $\trajDStratStar \in \failureSet$ for all $\usig \in \usigs$ by the theorem hypothesis.
But then $\xval \in \BAS$.
\end{proof}

\bibliographystyle{ieeetr}
\bibliography{references}

\end{document}